\documentclass[runningheads]{llncs}

\usepackage{amsmath}
\usepackage{amssymb}
\usepackage{listings}
\usepackage{graphicx}
\usepackage{multicol}
\usepackage{xcolor}

\newcommand{\elide}[1]{}

\newcommand{\compress}{\mathsf{compress}}

\newcommand{\trace}{\mathsf{trace}}
\newcommand{\ctrace}{\mathsf{ctrace}}     

\newcommand{\states}{\mathsf{states}}

\newcommand{\game}{\mathcal{G}}
\newcommand{\strat}{\mathcal{S}}
\newcommand{\zip}{\mathsf{zip}}
\newcommand{\unzip}{\mathsf{unzip}}

\newcommand{\Lang}{\mathcal{L}}

\title{Witnessing Secure Compilation}

\author{Kedar S. Namjoshi \inst{1}
  \and Lucas M. Tabajara \inst{2}}

\institute{
  Bell Labs, Nokia, Murray Hill, NJ, USA \email{kedar.namjoshi@nokia-bell-labs.com}
  \and
  Rice University, Houston, TX, USA \email{lucasmt@rice.edu}}

\begin{document}
\maketitle
\begin{abstract}
Compiler optimizations are designed to improve run-time performance while preserving input-output behavior. Correctness in this sense does not necessarily preserve security: it is known that standard optimizations may break or weaken security properties that hold of the source program. This work develops a translation validation method for secure compilation. Security (hyper-)properties are expressed using automata operating over a bundle of program traces. A flexible, automaton-based refinement scheme, generalizing existing refinement methods, guarantees that the associated security property is preserved by a program transformation. In practice, the refinement relations (``security witnesses'') can be generated during compilation and validated independently with a refinement checker. This process is illustrated for common optimizations. Crucially, it is not necessary to verify the compiler implementation itself, which is infeasible in practice for production compilers. 
\end{abstract}

\section{Introduction}

Optimizing compilers such as GCC and LLVM are used to improve run-time performance of software. Programmers expect the optimizing transformations to preserve program behavior. A number of approaches, ranging from automated testing (cf.~\cite{DBLP:conf/pldi/YangCER11,DBLP:conf/pldi/LeAS14}) to translation validation (cf.~\cite{CVT98,Necula00,ZuckPG03}) to full mathematical proof (cf.~\cite{DBLP:conf/popl/Leroy06}) have been utilized to check this property.

Programmers also implicitly expect that optimizations do not alter the security properties  of the source code (except, possibly, to strengthen them). It is surprising, then, to realize that even correctly implemented optimizations, such as the one illustrated in Figure~\ref{fig:dse}, may weaken security guarantees (cf.~\cite{scrub-2002,dsilva-et-al-2015}). 

\begin{figure}[h]
\begin{minipage}{0.4\textwidth}
\begin{verbatim}
  int x := read_secret_key();
  use(x);
  x := 0; // clear secret data
  rest_of_program();
\end{verbatim}
\end{minipage}
\hspace{0.2\textwidth}
\begin{minipage}{0.4\textwidth}
\begin{verbatim}
 int x := read_secret_key();
 use(x);
 skip; // dead store removed
 rest_of_program();
\end{verbatim}
\end{minipage}
\label{fig:dse}
\caption{Dead Store Elimination: Introducing Information Leaks}
\end{figure}

This common optimization removes stores that have no effect on the program's input-output behavior. Assuming that \verb|x| is not referenced in \verb|rest_of_program|, the transformation correctly replaces \verb|x := 0| in the source program on the left with \verb|skip|. The original assignment, however, was carefully placed to clear the memory holding the secret key. By removing it, the secret key becomes accessible in the rest of the program and vulnerable to attack. While this leak can be blocked using compiler directives~\cite{dce-usenix-2017}, such fixes are unsatisfactory, as they are not portable and might not block other vulnerabilities. 

The ideal is a mathematical guarantee that security is preserved across compiler optimizations. A full mathematical proof establishing this property over all source programs is technically challenging but feasible for a compact, formally developed compiler such as CompCert~\cite{DBLP:conf/popl/Leroy06,DBLP:conf/csfw/BartheGL18}. It is, however, infeasible for compilers such as GCC or LLVM with millions of lines of code written in hard-to-formalize languages such as \verb|C| and \verb|C++|.

The alternative method of \emph{Translation Validation}~\cite{CVT98} settles for the less ambitious goal of formally checking the result of every run of a compiler. In the form considered here, a compiler is \emph{designed} to generate additional information (called a ``certificate'' or a ``witness'') to simplify this check~\cite{rinard99,namjoshi-zuck-2013}. To establish correctness preservation, a refinement (i.e., simulation) relation is the natural choice for a witness. Moreover, such a relation is validated by checking inductiveness over a pair of single-step transitions (in the source and target programs), a condition that is easily encoded as an SMT query. 

For many optimizations, the associated refinement relations are logically simple\footnote{Examples include dead-store elimination, constant propagation, loop unrolling and peeling, loop-invariant code motion, static single assignment (SSA) conversion.}, being formed of equalities between variables at corresponding source and target program points~\cite{rinard99,darkothesis,namjoshi-zuck-2013}.  Crucially, neither the compiler nor the witness generation machinery have to be proved correct: an invalid witness points to either a flaw in compilation, or a misunderstanding by the compiler writer of the correctness argument; both outcomes are worth further investigation. 

In this work, we investigate translation validation for preservation of \emph{security} properties. Two key questions are: What is a useful witness format for  security properties? and How easy is it to generate and check witnesses for validity? Refinement proof systems are known only for two important security properties, non-interference~\cite{amorim-et-al-2014,deng-namjoshi-SAS-2016,DBLP:conf/eurosp/MurraySE18} and constant-time execution~\cite{DBLP:conf/csfw/BartheGL18}. 

In this context, a major contribution of this work is the development of a uniform, automaton-based refinement scheme for a large class of security properties. We use a logic akin to HyperLTL~\cite{DBLP:conf/post/ClarksonFKMRS14} to describe hyperproperties~\cite{DBLP:conf/sas/TerauchiA05,DBLP:conf/csfw/ClarksonS08} (which are sets of sets of sequences). A security property $\varphi$ is represented by a formula $Q_1 \pi_1, \ldots, Q_k \pi_k: \kappa(\pi_1,\ldots,\pi_k)$, where the $\pi_i$'s represent traces over an observation alphabet, the $Q_i$'s stand for either existential or universal quantifiers, and $\kappa$ is a set of bundles of $k$ program traces, represented by a non-deterministic B\"uchi automaton, $A$, whose language is the \emph{complement}\footnote{In this, we follow the standard practice in model checking of using automata for the negation of the desired property, rather than for the property itself.} of the set $\kappa$.

A transformation from program $S$ to program $T$ preserves a security property $\varphi$ if every violation of $\varphi$ by $T$ has a matching violation of $\varphi$ by $S$. Intuitively, matching violations have the same inputs and the same cause.

The first refinement scheme that we propose applies to purely universal properties, i.e., those of the form $\forall \pi_1 \ldots \forall \pi_k: \kappa(\pi_1,\ldots,\pi_k)$. The witness is a relation $R$ between the product transition systems $A \times T^k$ and $A \times S^k$ that meets certain inductiveness conditions. If $R$ can be defined, the transformation preserves $\varphi$. 

The second refinement scheme applies to properties with arbitrary quantification (the $\forall\exists$ alternation is needed to express limits on an attacker's knowledge). Here, the witness is a \emph{pair} of relations: one being a refinement relation (as before) between the product transition systems $A \times T^k$ and $A \times S^k$; the other component is an input-preserving bisimulation relation between $T$ and $S$.  We show that if such a pair of relations can be defined, the transformation from $S$ to $T$ preserves $\varphi$ in the sense above. 

We give examples of program transformations to illustrate the definition of security witnesses. As is the case for correctness, these witness relations also have a simple logical form, which can be analyzed by SMT solvers. The information needed to define these relations is present during compilation, so the relations may be easily generated by a compiler. 

Finally, we show that the refinement proof rules derived by this scheme from the automata-theoretic formulations of non-interference and constant-time are closely related to the known proof rules for these properties. Proofs carried out with the known notions are also valid for the newly derived ones. 

The key contribution of this work is in the flexible refinement scheme, which makes it possible to construct refinement proof rules for a wide range of security properties, including subtle variations on standard properties such as non-interference. The synthesized refinement rules may also be used for deductive proofs of security preservation. The primary inspiration for the automaton-based refinement notion comes from a beautiful paper by Manna and Pnueli~\cite{DBLP:conf/tls/MannaP87} demonstrating how to synthesize a deductive verification system for a temporal property from its associated automaton. 
  
\section{Example} \label{sec:example}

To illustrate the approach, consider the following source program, $S$.

\begin{verbatim}
L1: int x := read_secret_input();
L2: int y := 42;
L3: int z := y - 41;
L4: x := x * (z - 1);
\end{verbatim}

In this program, $\texttt{x}$ stores the value of a secret input. As will be described in Section~\ref{sec:transition-systems}, this program can be modeled as a transition system. The states of the system can be considered to be pairs $(\alpha, \ell)$. The first component $\alpha : \mathcal{V} \rightarrow \textsc{Int}$ is a partial assignment mapping variables in $\mathcal{V} = \{\texttt{x}, \texttt{y}, \texttt{z}\}$ to values in $\textsc{Int}$, the set of values that a variable of type $\texttt{int}$ can contain. The second component $\ell \in \textsc{Loc} = \{\texttt{L1}, \texttt{L2}, \texttt{L3}, \texttt{L4}\}$ is a location in the program, indicating the next instruction to be executed. In the initial state, $\alpha$ is empty and $\ell$ points to location $\texttt{L1}$. Transitions of the system update $\alpha$ according to the variable assignment instructions, and $\ell$ according to the control flow of the program.

To specify a notion of security for this program, two elements are necessary: an attack model describing what an attacker is assumed to be capable of observing (Section~\ref{sec:attack_models}) and a security property over a set of program executions (Section~\ref{sec:hyperproperties}). Suppose that an attacker can see the state of the memory at the end of the program, represented by the final value of $\alpha$, and the security property expresses that for every two possible executions of the program, the final state of the memory must be the same, regardless of the secret input, thus guaranteeing that the secret does not leak. Unlike correctness properties, this is a two-trace property, which can be written as a formula of the shape $\forall \pi_1 , \pi_2 : \kappa(\pi_1, \pi_2)$, where $\kappa(\pi_1, \pi_2)$ expresses that the memory at the end of the program is the same for traces $\pi_1$ and $\pi_2$ (cf. Section~\ref{sec:hyperproperties}). The negation of $\kappa$ can then be translated into an automaton $A$ that detects violations of this property.

It is not hard to see that the program satisfies the security property, since $\texttt{y}$ and $\texttt{z}$ have constant values and at the end of the program $\texttt{x}$ is $0$. However, it is important to make sure that this property is preserved after the compiler performs optimizations that modify the source code. This can be done if the compiler can provide a witness in the form of a \emph{refinement relation} (Section~\ref{sec:refinement}). Consider, for example, a compiler which performs constant folding, which simplifies expressions that can be inferred to be constant at compile time. The optimized program $T$ would be:

\begin{verbatim}
L1: int x := read_secret_input();
L2: int y := 42;
L3: int z := 1;
L4: x := 0;
\end{verbatim}

By taking the product of the automaton $A$ with two copies of $S$ or $T$ (one for each trace $\pi_i$ considered by $\kappa$), we obtain automata $A \times S^2$ and $A \times T^2$ whose language is the set of pairs of traces in each program that violates the property. Since this set is empty for $S$, it should be empty for $T$ as well, a fact which can be certified by providing a refinement relation $R$ between the state spaces of $A \times T^2$ and $A \times S^2$.

As the transformation considered here is very simple, the refinement relation is simple as well: it relates configurations $(q,t_0,t_1)$ and $(p,s_0,s_1)$ of the two spaces if the automaton states $p,q$ are identical, corresponding program states $t_0,s_0$ and $t_1,s_1$ are also identical (including program location), and the variables in $s_0$ and $s_1$ have the constant values derived at their location (see Section~\ref{sec:refinement-check} for details). The inductiveness of this relation over transitions of $A \times T^2$ and $A \times S^2$ can be easily checked by an SMT solver by representing the states symbolically.

\newcommand{\pstates}{C}
\newcommand{\pinputs}{I}
\newcommand{\poutputs}{O}
\newcommand{\pinit}{\iota}
\newcommand{\ptrans}{\rightarrow}
\newcommand{\PTRANS}{\Rightarrow}

\newcommand{\EXT}[1]{{#1}_{e}}

\newcommand{\pstate}{c}

\newcommand{\buff}[1]{\hat{#1}}

\newcommand{\emptystring}{\varepsilon}

\section{Background}

We propose an abstract program and attack model defined in terms of labeled transition systems. We also define B\"{u}chi automata over bundles of program traces, which will be used in the encoding of security properties, and describe a product operation between programs and automata that will assist in the verification of program transformations.

\subsubsection{Notation} Let $\Sigma$ be an \emph{alphabet}, i.e., a set of symbols, and let $\Gamma$ be a subset of $\Sigma$. An infinite sequence $u = u(0),u(1),\ldots$, where $u(i) \in \Sigma$ for all $i$, is said to be a ``sequence over $\Sigma$''. For variables $x,y$ denoting elements of $\Sigma$, the notation $x=_{\Gamma}y$ (read as ``$x$ and $y$ agree on $\Gamma$'') denotes the predicate where either $x$ and $y$ are both not in $\Gamma$, or $x$ and $y$ are both in $\Gamma$ and $x=y$. For a sequence $u$ over $\Sigma$, the notation $u|_{\Gamma}$ (read as ``$u$ projected to $\Gamma$'') denotes the sub-sequence of $u$ formed by elements in $\Gamma$. The operator $\compress(v) = v|_{\Sigma}$, applied to a  sequence $v$ over $\Sigma \cup \{\emptystring\}$, removes all $\emptystring$ symbols in $v$ to form a sequence over $\Sigma$. For a bundle of traces $w=(w_1,\ldots,w_k)$ where each trace is an infinite sequence of $\Sigma$, the operator $\zip(w)$ defines an infinite sequence over  $\Sigma^k$ obtained by choosing successive elements from each trace. In other words, $u=\zip(w)$ is defined by $u(i) = (w_1(i),\ldots,w_k(i))$, for all $i$. The operator $\unzip$ is its inverse.

\subsection{Programs as Transition Systems} \label{sec:transition-systems}

A program is represented as a transition system $S = (\pstates, \Sigma, \pinit, \ptrans)$:
\begin{itemize}
\item $\pstates$ is a set of program states, or configurations;
\item $\Sigma$ is a set of observables, partitioned into input, $\pinputs$, and output, $\poutputs$; 
\item $\pinit \in \pstates$ is the initial configuration;
\item $(\ptrans) \subseteq \pstates \times (\Sigma \cup \{\emptystring\}) \times \pstates$ is the transition relation.
\end{itemize}

Transitions labeled by input symbols in $\pinputs$ represent instructions in the program that read input values, while transitions labeled by output symbols in $\poutputs$ represent instructions that produce observable outputs. Transitions labeled by $\emptystring$ represent internal transitions where the state of the program changes without any observable effect.

An \emph{execution} is an infinite sequence of transitions $(\pstate_0, w_0, \pstate_1)(\pstate_1, w_1, \pstate_2) \ldots \in (\ptrans)^\omega$ such that $\pstate_0 = \pinit$ and adjacent transitions are connected as shown. (We may write this as the alternating sequence $\pstate_0,w_0,\pstate_1,w_1,\pstate_2,\ldots$.) To ensure that every execution is infinite, we assume that $(\ptrans)$ is left-total. To model programs with finite executions, we assume that the alphabet has a special termination symbol $\bot$, and add a transition $(\pstate, \bot, \pstate)$ for every final state $\pstate$. We also assume that there is no infinite execution where the transition labels are always $\emptystring$ from some point on.

An execution $x=(\pstate_0, w_0, \pstate_1)(\pstate_1, w_1, \pstate_2) \ldots$  has an associated \emph{trace}, denoted $\trace(x)$, given by the sequence $w_0, w_1, \ldots$ over $\Sigma \cup \{\emptystring\}$. The compressed trace of execution $x$, $\compress(\trace(x))$, is denoted $\ctrace(x)$. The final assumption above ensures that the compressed trace of an infinite execution is also infinite. The sequence of states on an execution $x$ is denoted by $\states(x)$.

\subsection{Attack Models as Extended Transition Systems} \label{sec:attack_models}
The choice of how to model a program as a transition system depends on the properties one would like to verify. For correctness, it is enough to use the standard input-output semantics of the program. To represent security properties, however, it is usually necessary to extend this base semantics to bring out interesting features. Such an extension typically adds auxiliary state and new observations needed to model an attack. For example, if an attack is based on program location, that is added as an auxiliary state component in the extended program semantics. Other examples of such structures are modeling a program stack as an array with a stack pointer, explicitly tracking the addresses of memory reads and writes, and distinguishing between cache and main memory accesses. These extended semantics are roughly analogous to the \emph{leakage models} of~\cite{DBLP:conf/csfw/BartheGL18}. The base transition system is extended to one with a new state space, denoted $\EXT{\pstates}$; new observations, denoted $\EXT{\poutputs}$; and a new alphabet, $\EXT{\Sigma}$, which is the union of $\Sigma$ with $\EXT{\poutputs}$. The extensions do not alter input-output behavior; formally, the original and extended systems are bisimular with respect to $\Sigma$.

\subsection{B\"uchi Automata over trace bundles}
A B\"{u}chi automaton over a bundle of $k$ infinite traces over $\EXT{\Sigma}$ is specified as $A = (Q, \EXT{\Sigma}^k, \iota, \Delta, F)$, where:
\begin{itemize}
\item $Q$ is the state space of the automaton;
\item $\EXT{\Sigma}^k$ is the alphabet of the automaton, each element is a $k$-vector over $\EXT{\Sigma}$;
\item $\iota \in Q$ is the initial state;
\item $\Delta \subseteq Q \times \EXT{\Sigma}^k \times Q$ is the transition relation;
\item $F \subseteq Q$ is the set of accepting states.
\end{itemize}

A \emph{run} of $A$ over a bundle of traces $t=(t_1,\ldots,t_k) \in (\Sigma^\omega)^k$ is an alternating sequence of states and symbols, of the form $(q_0=\iota),a_0,q_1,a_1,q_2,\ldots$ where for each $i$, $a_i = (t_1(i), \ldots, t_k(i))$  --- that is, $a_0,a_1,\ldots$ equals $\zip(t)$ --- and $(q_i,a_i,q_{i+1})$ is in the transition relation $\Delta$.  The run is accepting if a state in $F$ occurs infinitely often along it. The \emph{language} accepted by $A$, denoted by $\mathcal{L}(A)$, is the set of all $k$-trace bundles that are accepted by $A$.

\subsubsection{Automaton-Program Product}
In verification, the set of traces of a program that violate a property can be represented by an automaton that is the product of the program with an automaton for the negation of that property. Security properties may require analyzing multiple traces of a program; therefore, we define the analogous automaton as a product between an automaton $A$ for the negation of the security property and the $k$-fold composition $P^k$ of a program $P$. For simplicity, assume for now that the program $P$ contains no $\emptystring$-transitions.
Programs with $\emptystring$-transitions can be handled by converting $A$ over $\EXT{\Sigma}^k$ into a new automaton $\buff{A}$ over $(\EXT{\Sigma} \cup \{\emptystring\})^k$ (see Appendix for details).

Let $A = (Q^A, \EXT{\Sigma}^k, \Delta^A, \iota^A, F^A)$ be a B\"uchi automaton (over a $k$-trace bundle) and $P=(\pstates,\EXT{\Sigma},\pinit,\ptrans)$ be a program. The product of $A$ and $P^k$, written $A \times P^k$, is a B\"uchi automaton $B=(Q^B, \EXT{\Sigma}^k, \Delta^B, \iota^B, F^B)$, where:
\begin{itemize}
\item $Q^B=Q^A \times \pstates^k$;
\item $\iota^B = (\iota^A,(\pinit,\ldots,\pinit))$;
\item $((q,s),u,(q',s'))$ is in $\Delta^B$ if, and only if, $(q,u,q')$ is in $\Delta^A$, and $(s_i,u_i,s'_i)$ is in $(\ptrans)$ for all $i$;
\item $(q,s)$ is in $F^B$ iff $q$ is in $F^A$.
\end{itemize}

\begin{lemma} \label{lemma:product}
  Trace $\zip(t_1, \ldots, t_k)$ is in $\Lang(A \times P^k)$ if, and only if, $\zip(t_1, \ldots, t_k)$ is in $\Lang(A)$ and, for all $i$, $t_i=\trace(x_i)$ for some execution $x_i$ of $P$.
\end{lemma}

\subsubsection{Bisimulations} \label{sec:bisimulation}
For programs $S = (\pstates^S, \EXT{\Sigma}, \pinit^S, \ptrans^S)$ and $T = (\pstates^T, \EXT{\Sigma}, \pinit^T, \ptrans^T)$, and a subset $I$ of $\EXT{\Sigma}$, a relation $B \subseteq \pstates^T \times \pstates^S$ is a \emph{bisimulation} for $I$ if: 
\begin{enumerate}
\item $(\pinit^T, \pinit^S) \in B$;

\item For every $(t,s)$ in $B$ and $(t,v,t')$ in $(\ptrans^T)$ there is $u$ and $s'$ such that $(s,u,s')$ is in $(\ptrans^S)$ and $(t',s') \in B$ and $u =_I v$.

\item For every $(t,s)$ in $B$ and $(s,u,s')$ in $(\ptrans^S)$ there is $v$ and $t'$ such that $(t,v,t')$ is in $(\ptrans^T)$ and $(t',s') \in B$ and $u =_I v$.
\end{enumerate}

\section{Formulating Security Preservation} \label{sec:hyperproperties}

A temporal correctness property $\varphi$ is expressed as a set of infinite traces. Many security properties can only be described as properties of pairs or tuples of traces. A standard example is that of \emph{noninterference}, which  models potential leakage of secret inputs: if two program traces differ only in secret inputs, they should be indistinguishable to an observer that can only view non-secret inputs and outputs. The general notion is that of  a \emph{hyperproperty} \cite{DBLP:conf/sas/TerauchiA05,DBLP:conf/csfw/ClarksonS08}, which is a set containing sets of infinite traces; a program satisfies a hyperproperty $H$ if the set of all compressed traces of the program is an element of $H$. Linear Temporal Logic (LTL) is commonly used to express correctness properties. Our formulation of security properties is an extension of the logic HyperLTL, which can express common security properties including several variants of noninterference~\cite{DBLP:conf/post/ClarksonFKMRS14}.

A security property $\varphi$ has the form $(Q_1 \pi_1, \ldots, Q_n \pi_k: \kappa(\pi_1,\ldots,\pi_k))$, where the $Q_i$'s are first-order quantifiers over trace variables, and $\kappa$ is set of $k$-trace bundles, described by a B\"uchi automaton whose language is the \emph{complement} of $\kappa$. This formulation borrows the crucial notion of trace quantification from HyperLTL, while generalizing it, as automata are more expressive than LTL, and atomic propositions may hold of $k$-vectors rather than on a single trace.

The satisfaction of property $\varphi$ by a program $P$ is defined in terms of the following finite two-player game, denoted $\game(P,\varphi)$. The protagonist, Alice, chooses an execution of $P$ for each existential quantifier position, while the antagonist, Bob, chooses an execution of $P$ at each universal quantifier position. The choices are made in sequence, from the outermost to the innermost quantifier. A play of this game is a maximal sequence of choices. The outcome of a play is thus a ``bundle'' of program executions, say $\sigma = (\sigma_1,\ldots,\sigma_k)$. This induces a corresponding bundle of compressed traces, $t = (t_1,\ldots,t_k)$, where $t_i=\ctrace(\sigma_i)$ for each $i$. This play is a win for Alice if $t$ satisfies $\kappa$ and a win for Bob otherwise.

A \emph{strategy} for Bob is a function, say $\xi$, that defines a non-empty set of executions for positions $i$ where $Q_i$ is a universal quantifier, in terms of the earlier choices $\sigma_1,\ldots,\sigma_{i-1}$; the choice of $\sigma_i$ is from this set.
A strategy for Alice is defined symmetrically. A strategy is \emph{winning} for player $X$ if every play following the strategy is a win for $X$. This game is determined, in that for any program $P$ one of the players has a winning strategy. Satisfaction of a security property is defined by the following. 

\begin{definition}
  Program $P$ satisfies a security property $\varphi$, written $\models_P \varphi$, if the protagonist has a winning strategy in the game $\game(P,\varphi)$.
\end{definition}

\subsection{Secure Program Transformation} \label{sec:secure_transformation}

Let $S = (\pstates^S, \EXT{\Sigma}, \pinit^S, \ptrans^S)$ be the transition system representing the original \emph{source} program and let $T = (\pstates^T, \EXT{\Sigma}, \pinit^T, \ptrans^T)$ be the transition system for the transformed \emph{target} program. Any notion of secure transformation must imply the preservation property that if $S$ satisfies $\varphi$ and the transformation from $S$ to $T$ is secure for $\varphi$ then $T$ also satisfies $\varphi$. This property in itself is, however, too weak to serve as a definition of secure transformation.

Consider the transformation shown in Figure~\ref{fig:dse}, with \verb|use(x)| defined so that it terminates execution if the secret key \verb|x| is invalid. As the source program violates non-interference by leaking the validity of the key, the transformation would be trivially secure if the preservation property is taken as the definition of secure transformation. But that conclusion is wrong: the leak introduced in the target program is clearly different and of a more serious nature, as the entire secret key is now vulnerable to attack.

This analysis prompts the formulation of a stronger principle for secure transformation. (Similar principles have been discussed in the literature, e.g.,~\cite{DBLP:conf/ccs/FournetGR09}.) The intuition is that every instance and type of violation in $T$ should have a matching instance and type of violation in $S$. To represent different types of violations, we suppose that the negated property is represented by a collection of automata, each checking for a specific type of violation. 

\begin{definition}
A strategy $\xi^S$ for the antagonist in $\game(S,\varphi)$ (representing a violation in $S$) \emph{matches} a strategy $\xi^T$ for the antagonist in game $\game(T,\varphi)$ (representing a violation in $T$) if for every maximal play $u=u_1,\ldots,u_k$ following $\xi^T$, there is a maximal play $v=v_1,\ldots,v_k$ following $\xi^S$ such that (1) the two plays are input-equivalent, i.e., $u_i|_{\pinputs} = v_i|_{\pinputs}$ for all $i$, and (2) if $u$ is accepted by the $m$-th automaton for the negated property, then $v$ is accepted by the same automaton.
\end{definition}

\begin{definition}\label{def:preservation}
  A transformation from $S$ to $T$ preserves security property $\varphi$ if for every winning strategy for the antagonist in the game $\game(T,\varphi)$, there is a matching winning strategy for the antagonist  in the game $\game(S,\varphi)$. 
\end{definition}

As an immediate consequence, we have the preservation property. 
\begin{theorem}
  If a transformation from $S$ to $T$ preserves security property $\varphi$ and if $S$ satisfies $\varphi$, then $T$ satisfies $\varphi$.
\end{theorem}

In the important case where the security property is purely universal, of the form $\forall \pi_1, \ldots, \forall \pi_k : \kappa(\pi_1,\ldots,\pi_k)$, a winning strategy for the antagonist is simply a bundle of $k$ traces, representing an assignment to $\pi_1, \ldots, \pi_k$ that falsifies $\kappa$.

\section{Refinement for Preservation of Universal Properties} \label{sec:refinement}

We define an automaton-based refinement scheme that is sound for purely-universal properties  $\varphi$, of the form $(\forall \pi_1, \ldots, \forall \pi_k: \kappa(\pi_1,\ldots,\pi_k))$. Section~\ref{sec:general-witness} generalizes this to properties with arbitrary quantifier prefixes. We assume for simplicity that programs $S$ and $T$ have no $\emptystring$-transitions; we discuss how to remove this assumption at the end of the section. An automaton-based refinement scheme for preservation of $\varphi$ is defined below.

\begin{definition} \label{def:refinement}
Let $S,T$ be programs over the same alphabet, $\EXT{\Sigma}$, and $A$ be a B\"{u}chi automaton over $\EXT{\Sigma}^k$. Let $\pinputs$ be a subset of $\EXT{\Sigma}$. A relation $R \subseteq (Q^A \times (\pstates^T)^k)  \times (Q^A \times (\pstates^S)^k)$ is a \emph{refinement relation} from $A \times T^k$ to $A \times S^k$ for $\pinputs$ if 
\begin{enumerate}
\item Initial configurations are related, i.e., $((\iota^A, \pinit^{T^k}), (\iota^A,\pinit^{S^k}))$ is in $R$, and 
\item Related states have matching transitions. That is, if $((q,t),(p,s)) \in R$ and $((q,t), v, (q',t')) \in$$\Delta^{A\times T^k}$, there are $u$, $p'$, and $s'$ such that the following hold:
  \begin{enumerate}
  \item $((p,s),u,(p',s'))$ is a transition in $\Delta^{A \times S^k}$;
  \item $u$ and $v$ agree on $\pinputs$, that is, $u_i =_\pinputs v_i$ for all $i$; 
  \item the successor configurations are related, i.e., $((q',t'), (p',s')) \in R$; and
  \item acceptance is preserved, i.e., if $q' \in F$ then $p' \in F$.
  \end{enumerate}
\end{enumerate}
\end{definition}

\begin{lemma} \label{lemma:typeI-refinement}
  If there exists a refinement from $A \times T^k$ to $A \times S^k$ then, for every sequence $v$ in $\Lang(A \times T^k)$, there is a  sequence $u$ in $\Lang(A \times S^k)$ such that $u$ and $v$ are input-equivalent.
\end{lemma}

\begin{theorem}[Universal Refinement] \label{thrm:refinement}
Let $\varphi = (\forall \pi_1, \ldots, \pi_k: \kappa(\pi_1,\ldots,\pi_k))$ be a universal security property; $S$ and $T$ be programs over a common alphabet $\EXT{\Sigma} = \Sigma \cup \EXT{\poutputs}$; $A = (Q, \EXT{\Sigma}^k, \iota, \Delta, F)$ be an automaton for the negation of $\kappa$; and $R \subseteq (Q \times (\pstates^T)^k) \times (Q \times (\pstates^S)^k)$ be a refinement relation from $A \times T^k$ to $A \times S^k$ for $\pinputs$. Then, the transformation from $S$ to $T$ preserves $\varphi$. 
\end{theorem}

\begin{proof}
  A violation of $\varphi$ by $T$ is given by a bundle of executions of $T$ that violates $\kappa$. We show that there is an input-equivalent bundle of executions of $S$ that also violates $\kappa$. Let $x = (x_1,\ldots,x_k)$ be a bundle of executions of $T$ that does not satisfy $\kappa$. By Lemma~\ref{lemma:product}, $v = \zip(\trace(x_1),\ldots,\trace(x_k))$ is accepted by $A \times T^k$. By Lemma~\ref{lemma:typeI-refinement}, there is a sequence $u$ accepted by $A \times S^k$ that is input-equivalent to $v$. Again by Lemma~\ref{lemma:product}, there is a bundle of executions $y = (y_1,\ldots,y_k)$ of $S$ such that $u = \zip(\trace(y_1),\ldots,\trace(y_k))$ and $y$ violates $\kappa$. As $u$ and $v$ are input equivalent, $\trace(x_i)$ and $\trace(y_i)$ are input-equivalent for all $i$, as required.
\qed
\end{proof}

The refinement proof rule for universal properties is implicit: a witness is a relation $R$ from $A \times T^k$ to $A \times S^k$; this is valid if it satisfies the conditions set out in Definition~\ref{def:refinement}. The theorem establishes the soundness of this proof rule. Examples of witnesses for specific compiler transformations are given in Section~\ref{sec:refinement-check}, which also discusses SMT-based checking of the proof requirements.

To handle programs that include $\emptystring$-transitions, we can convert the automaton $A$ over $\EXT{\Sigma}^k$ into a \emph{buffering automaton} $\buff{A}$ over $(\EXT{\Sigma} \cup \{\emptystring\})^k$, such that $\buff{A}$ accepts $\zip(v_1,\ldots,v_k)$ iff $A$ accepts $\zip(\compress(v_1),\ldots,\compress(v_k))$. The refinement is then defined over $\buff{A} \times S^k$ and $\buff{A} \times T^k$. Details can be found in the Appendix. Another useful extension is the addition of \emph{stuttering}, which can be necessary for example when a transformation removes instructions. Stuttering relaxes Definition~\ref{def:refinement} to allow multiple transitions on the source to match a single transition on the target, or vice-versa. This is a standard technique for verification~\cite{DBLP:journals/tcs/BrowneCG88} and one-step formulations suitable for SMT solvers are known (cf.~\cite{DBLP:conf/fsttcs/Namjoshi97,DBLP:conf/popl/Leroy06}).

\section{Checking Transformation Security} \label{sec:refinement-check}

This section first describes how to construct the SMT formula that checks the correctness of a given refinement relation. Next, it demonstrates through concrete examples how to express a refinement relation for specific program transformations.

\subsection{Refinement Check} \label{sec:smt-check}

Assume that the refinement relation $R$, the transition relations $\Delta$, $(\ptrans_T)$ and $(\ptrans_S)$ and the set of accepting states $F$ are described by SMT formulas over variables ranging over states and alphabet symbols.

To verify that the formula $R$ is indeed a refinement, we perform an inductive check following the definition of refinement given in Definition~\ref{def:refinement}. To prove the base case, which says that the initial states of $A \times T^k$ and $A \times S^k$ are related by $R$, we simply evaluate the formula on the initial states.

Proving the inductive step again follows from Definition~\ref{def:refinement}, which states that the transition relations of the automata must preserve membership in $R$. The correctness of the inductive step would be expressed by an SMT query of the shape $(\forall q^T, q^S, p^T,  t, s, t', \sigma^T:  (\exists \sigma^S, p^S, s' : \varphi_1 \rightarrow \varphi_2))$, where:
\begin{alignat*}{2}
&\varphi_1 \equiv
&&R((q^T, t), (q^S, s)) \land \Delta(q^T, \sigma^T, p^T) \land \bigwedge^k_{i=1} (t_i \xrightarrow{\sigma^T_i}_T t'_i) \\
&\varphi_2 \equiv
&&\Delta(q^S, \sigma^S, p^S) \land \bigwedge^k_{i=1} (s_i \xrightarrow{\sigma^S_i}_S s'_i) \land \bigwedge^k_{i=1} (\sigma^T_i =_\pinputs \sigma^S_i) \\
& &&\land R((p^T, t'), (p^S, s')) \land (F(p^T) \rightarrow F(p^S))
\end{alignat*}

Note that this formula has a quantifier alternation, which is hard for SMT solvers to handle. However, the formula can be reduced to a validity check by providing Skolem functions from the universal to the existential variables. We expect the compiler to provide these functions. As we will see in the examples below, in many cases the compiler can choose simple-enough Skolem functions that the validity of the formula can be verified using only equality reasoning, making it unnecessary to even expand $\Delta$ and $F$ to their definitions. More generally, a compiler writer must have a proof in mind for each optimization and should therefore be able to provide the necessary Skolem functions to match the refinement relation.

\subsection{Refinement Relations for Compiler Optimizations}

We consider three common optimizations below. In addition, further examples for \emph{dead-branch elimination}, \emph{expression flattening}, \emph{loop peeling} and \emph{register spilling} can be found in the Appendix. All transformations were based on the examples in~\cite{DBLP:conf/csfw/BartheGL18}.

\subsubsection{Example 1: Constant Folding} \label{sec:constant-folding}

Section~\ref{sec:example} presented an example of a program transformation by constant folding. We now proceed to show how a refinement relation can be defined to serve as a witness for the security of this transformation, so its validity can be checked using an SMT solver as described above.

Recall that states of $S$ and $T$ are of the form $(\alpha, \ell)$, where $\alpha : \mathcal{V} \rightarrow \textsc{Int}$ and $\ell \in \textsc{Loc}$.
Then, $R$ can be expressed by the following formula over states $q^T, q^S$ of the automaton $A$ and states $t$ of $T^k$ and $s$ of $S^k$, where $t_i = (\alpha^T_i, \ell^T_i)$:
\begin{align*} \label{eq:refinement}
&(q^T = q^S) \land (t = s)
\land \bigwedge^k_{i = 1} (\ell^T_i = \texttt{L3} \rightarrow \alpha^T_i(\texttt{y}) = 42) \\
&\land \bigwedge^k_{i = 1} (\ell^T_i = \texttt{L4} \rightarrow \alpha^T_i(\texttt{z}) = 1) \land \bigwedge^k_{i = 1} (\ell^T_i = \texttt{L5} \rightarrow \alpha^T_i(\texttt{x}) = 0)
\end{align*}

Since this is a simple transformation, equality between states is all that is needed to establish a refinement. However, to allow the refinement to be verified automatically, the relation also has to carry information about the constant values at specific points in the program. In general, if the transformation relies on the fact that at location $\ell$ variable $v$ has constant value $c$, the constraint $\bigwedge^k_{i = 1} (\ell^T_i = \ell \rightarrow \alpha^T_i(v) = c)$ is added to $R$.

$R$ can be checked using the SMT query described in Section~\ref{sec:smt-check}. Note that for this particular transformation the compiler can choose Skolem functions that assign $\sigma^S = \sigma^T$ and $p^S = p^T$. In this case, from $(q^T = q^S)$ (given by the refinement relation) and $\Delta(q^T, \sigma^T, p^T)$ the solver can automatically infer $\Delta(q^S, \sigma^S, p^S)$, $(\sigma^T_i =_\pinputs \sigma^S_i)$ and $F(p^T) \rightarrow F(p^S)$ using only equality reasoning. Therefore, the refinement check is \emph{independent} of the security property in this case. This applies to many other transformations that occur in practice as well, since their reasons for preserving security are usually simple.

\subsubsection{Example 2: Common-Branch Factorization}

Common-branch factorization is a program optimization applied to conditional blocks where the instructions at the beginning of the \emph{then} and \emph{else} blocks are the same. If the condition does not depend on a variable modified by the common instruction, this instruction can be moved outside of the conditional.
Consider for example:

\begin{multicols}{2}
\begin{verbatim}
// Source program S
L1: if (j < arr_size) {
L2:     a := arr[0];
L3:     b := arr[j];
L4: } else {
L5:     a := arr[0];
L6:     b := arr[arr_size - 1];
L7: }
\end{verbatim}

\columnbreak

\begin{verbatim}
// Target program T
L1: a := arr[0];
L2: if (j < arr_size) {
L3:     b := arr[j];
L4: } else {
L5:
L6:     b := arr[arr_size - 1];
L7: }
\end{verbatim}
\end{multicols}

Suppose that the attack model allows the attacker to observe memory accesses, represented by the index \texttt{j} of every array access \texttt{arr[j]}. We assume that other variables are stored in registers rather than memory (see Appendix for a discussion on register spilling). Under this attack model the compressed traces produced by $T$ are identical to the ones of $S$, therefore the transformation will be secure regardless of the security property $\varphi$. However, because the order of instructions is different, a more complex refinement relation $R$ will be needed, compared to constant folding:
\begin{align*}
((t = s) \land (q^T = q^S)) \lor &\bigwedge^k_{i=1}((\ell^T_i = \texttt{L2}) \\
&\land ((\alpha^S_i(\texttt{i}) < \alpha^S_i(\texttt{arr\_size}))\,?\,(\ell^S_i = \texttt{L2}) : (\ell^S_i = \texttt{L5})) \\
&\land (\alpha^T_i = \alpha^S_i[\texttt{a} := \texttt{arr[0]}])) \land \Delta(q^S, (0, \ldots, 0), q^T)
\end{align*}

The refinement relation above expresses that the states of the programs and the automata are identical except when $T$ has executed the factored-out instruction but $S$ hasn't yet. At that point, $T$ is at location $L2$ and $S$ is either at location $L2$ or $L5$, depending on how the guard was evaluated. Note that it is necessary for $R$ to know that the location of $S$ depends on the evaluation of the guard, so that it can verify that at the next step $T$ will follow the same branch. The states of $\buff{A} \times S^k$ and $\buff{A} \times T^k$ are then related by saying that after updating $\texttt{a := arr[0]}$ on every track of $S$ the two states will be the same. Note that since this instruction produces an observation representing the index of the array access, the states of the automata are related by $\Delta(q^S, (0, \ldots, 0), q^T)$, indicating that the access has been observed by $\buff{A} \times T^k$ but not yet by $\buff{A} \times S^k$.

\subsubsection{Example 3: Switching Instructions}

This optimization switches two sequential instructions if the compiler can guarantee that the program's behavior will not change.
For example, consider the following source and target programs:

\begin{multicols}{2}
\begin{verbatim}
// Source program S
L1: int a[10], b[10];
L2: a[0] := secret_input();
L3: b[0] := secret_input();
L4: for(int j:=1; j<10; j++){
L5:   a[j] := b[j-1];
L6:   b[j] := a[j-1];
L7:   public_output(j);
L8: }
\end{verbatim}

\columnbreak

\begin{verbatim}
  // Target program T
  L1: int a[10], b[10];
  L2: a[0] := secret_input();
  L3: b[0] := secret_input();
  L4: for(int j:=1; j<10; j++){
  L5:   b[j] := a[j-1];
  L6:   a[j] := b[j-1];
  L7:   public_output(j);
  L8: }
\end{verbatim}
\end{multicols}

Note that the traces produced by $T$ and $S$ are identical. Therefore, a refinement relation for this pair of programs can be given by the following formula, regardless of the security property under verification:
\begin{align*}
    (q^S = q^T) \land & \bigwedge^k_{i=1} (\ell^S_i = \ell^T_i)
     \land (\ell^S_i \neq \texttt{L6} \rightarrow \alpha^S_i = \alpha^T_i) \\
    & \land (\ell^S_i = \texttt{L6} \rightarrow \alpha^S_i[\texttt{b[j]} := \texttt{a[j-1]}] = \alpha^T_i[\texttt{a[j]} := \texttt{b[j-1]}])
\end{align*}

The formula expresses that the state of the source and target programs is the same except between executing the two switched instructions. At that point, the state of the two programs is related by saying that after executing the second instruction in each of the programs they will again have the same state.

More generally, a similar refinement relation can be used for any source-target pair that satisfies the assumptions that (a) neither of the switched instructions produces an observable output, and (b) after both switched instructions are executed, the state of the two programs is always the same. All that is necessary in this case is to replace $\texttt{L6}$ by the appropriate location $\ell^S_{switch}$ where the switch happens and $\alpha^S_i[\texttt{b[j]} := \texttt{a[j-1]}] = \alpha^T_i[\texttt{a[j]} := \texttt{b[j-1]}]$ by an appropriate formula $\delta(\alpha^S_i, \alpha^T_i)$ describing the relationship between the states of the two programs at that location.

If the instructions being switched do produce observations, setting up the refinement relation becomes harder. This is due to the fact that the relationship $(q^S = q^T)$ might not hold in location $\ell^S_{switch}$, but expressing the true relationship between $q^S$ and $q^T$ is complex and might require knowledge of the state of all copies of $S$ and $T$ at once. It becomes simpler for some special cases, for example if the different copies of $S$ and $T$ are guaranteed to be synchronized (i.e., it is always the case that $\ell^S_i = \ell^S_j$ and $\ell^T_i = \ell^T_j$, for all $i$ and $j$), or if only one of the instructions produces an observation. Details can be found in the Appendix.

\section{Connections to Existing Proof Rules}

We establish connections to known proof rules for preservation of the non-interference \cite{amorim-et-al-2014,deng-namjoshi-SAS-2016,DBLP:conf/eurosp/MurraySE18} and constant-time~\cite{DBLP:conf/csfw/BartheGL18} properties. We show that under the assumptions of those rules, there is a simple and direct definition of a relation that meets the automaton-based refinement conditions for automata representing these properties. The automaton-based refinement method is thus general enough to serve as a uniform replacement for the specific proof methods.

\subsection{Constant Time}

We first consider the lockstep CT-simulation proof rule introduced in~\cite{DBLP:conf/csfw/BartheGL18} to show preservation of the constant-time property. For lack of space, we refer the reader to the original paper for the precise definitions of observational non-interference (Definition 1), constant-time as observational non-interference (Definition 4), lockstep simulation (Definition 5, denoted $\approx$), and lockstep CT-simulation (Definition 6, denoted $(\equiv_S,\equiv_C)$).

We do make two minor adjustments to better fit the automaton notion, which is based on trace rather than state properties. First, we add a dummy initial source state $\hat{S}(i)$ with a transition with input label $i$ to the actual initial state $S(i)$; and similarly for the target program, $C$. Secondly, we assume that a final state has a self-loop with a special transition label, $\bot$. Then the condition $(b \in S_f \leftrightarrow b' \in S_f)$ from Definition 1 in~\cite{DBLP:conf/csfw/BartheGL18} is covered by the (existing) label equality $t=t'$. With these changes, the observational non-interference property can be represented in negated form by the  automaton shown in Figure~\ref{fig:ct-automaton}, which simply looks for a sequence starting with an initial pair of input values satisfying $\phi$ and ending in unequal transition labels. The states are $I$ (initial), $S$ (sink), $M$ (mid), and $F$ (fail), which is also the accepting state. 

\begin{figure}
  \begin{center}
    \includegraphics[scale=0.7]{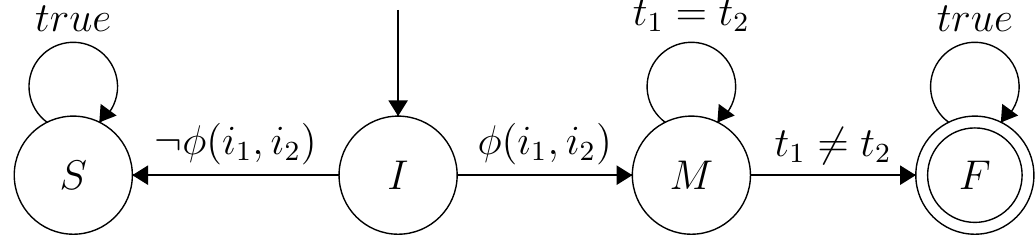}
  \end{center}
  \caption{A B\"uchi automaton for the negation of the constant-time property.}
  \label{fig:ct-automaton}
\end{figure}

We now define the automaton-based relation, using the notation in Theorem 1 of~\cite{DBLP:conf/csfw/BartheGL18}. Define relation $R$ by $(q,\alpha,\alpha') R (p,a,a')$ if $a \approx \alpha$, $a' \approx \alpha'$, and 
\begin{enumerate}
\item $p=F$, i.e., $p$ is the fail state, or
\item $p=q=S$, or
\item $p=q=I$, and $\alpha=\hat{C}(i)$, $\alpha'=\hat{C}(i')$, $a=\hat{S}(i)$, $a=\hat{S}(i')$, for some $i,i'$, or
\item $p=q=M$, and  $\alpha \equiv_C \alpha'$, and $a \equiv_S a'$.
\end{enumerate}

\begin{theorem}\label{thm:ct-simulation}
  If $(\equiv_S,\equiv_C)$ is a lockstep CT-simulation with respect  to the lockstep simulation $\approx$, the relation $R$ is a valid refinement relation.
\end{theorem}

\begin{proof}

  Every initial state of $\mathcal{A} \times C^2$ has a related initial state in $\mathcal{A} \times S^2$. 
  As related configurations are pairwise connected by $\approx$, which is a simulation, it follows that any pairwise transition from a $C$-configuration is matched by a pairwise transition from the related $S$-configuration, producing states $b,b'$ and  $\beta,\beta'$ that are pointwise related by $\approx$. These transitions have identical input labels, as the only transitions with input labels are those from the dummy initial states. 

  The remaining question is whether the successor configurations are connected by $R$. We reason by cases. 

  First, if $p=F$, then $p'$ is also $F$. Hence, the successor configurations are related. This is also true of the second condition, where $p=q=S$, as the successor states are $p'=q'=S$.

  If $p=q=I$ the successor states are $\beta=C(i),\beta'=C(i')$ and $b=S(i),b'=S(i')$, and the successor automaton state is either $p'=q'=S$, if $\phi(i,i')$ does not hold, or $p'=q'=M$, if it does. In the first possibility, the successor configurations are related by the second condition; in the second, they are related by the final condition, as  $C(i) \equiv_C C(i')$ and $S(i) \equiv_S S(i')$ hold if $\phi(i,i')$ does [Definition 6 of~\cite{DBLP:conf/csfw/BartheGL18}].

  Finally, consider the interesting case where $p=q=M$. Let $\tau,\tau'$ be the transition labels on the pairwise transition in $C$, and let $t,t'$ be the labels on the corresponding pairwise transition in $S$. We consider two cases:

  (1) Suppose $t \neq t'$. Then $p'=F$ and the successor configurations are related, regardless of $p'$.

  (2) Otherwise, $t=t'$ and $p'=M$. By CT-simulation [Definition 6 of of~\cite{DBLP:conf/csfw/BartheGL18}: $a\equiv_S a'$ and $\alpha \equiv_C \alpha'$ by the relation $R$], it follows that $b \equiv_S b'$ and $\beta \equiv_C \beta'$ hold, and $\tau=\tau'$. Thus, the successor automaton state on the $C$-side is $q'=M$ and the successor configurations are related by the final condition.
  
  This completes the case analysis. Finally, the definition of $R$ implies that if $q=F$ then $p=F$, as required.

\qed
\end{proof}

\subsection{Non-Interference}
Refinement-based proof rules for preservation of non-interference have been introduced in~\cite{amorim-et-al-2014,deng-namjoshi-SAS-2016,DBLP:conf/eurosp/MurraySE18}. The rules are not identical but are substantially similar in nature, providing conditions under which an ordinary simulation relation, $\prec$, between programs $C$ and $S$ implies preservation of non-interference. We choose the rule from~\cite{deng-namjoshi-SAS-2016}, which requires, in addition to the requirement that $\prec$ is a simulation preserving input and output events, that  (a) A final state of $C$ is related by $\prec$ only to a final state of $S$ (precisely, both are final or both non-final), and (b)  If $t_0 \prec s_0$ and $t_1 \prec s_1$ hold, and all states are either initial or final, then the low variables of $t_0$ and $t_1$ are equal iff the low variables of $s_0$ and $s_1$ are equal.

We make two minor adjustments to better fit the automaton notion, which is based on trace rather than state properties. First, we add a dummy initial source state $\hat{S}(i)$ with a transition that exposes the value of local variables and moves to the actual initial state $S(i)$ ($i$ is the secret input); and similarly for the target program, $C$. Secondly, we assume that a final state has a self-loop with a special transition label that exposes the value of local variables on termination.  With these changes, the negated non-interference can be represented by the  automaton shown in Figure~\ref{fig:non-interference}. It accepts an pair of execution traces if, and only if, initially the low-variables on the two traces have identical values, and either the corresponding outputs differ at some point, or final values of the low-variables are different. (The transition conditions are written as Boolean predicates which is a readable notation for describing a set of pairs of events; e.g., the $Low_1 \neq Low_2$ transition from state $I$ represents the set of pairs $(a,b)$ where $a$ is the $init(Low=i)$ event, $b$ is the $init(Low=j)$ event, and $i\neq j$.)

\begin{figure}
  \begin{center}
    \includegraphics[scale=0.7]{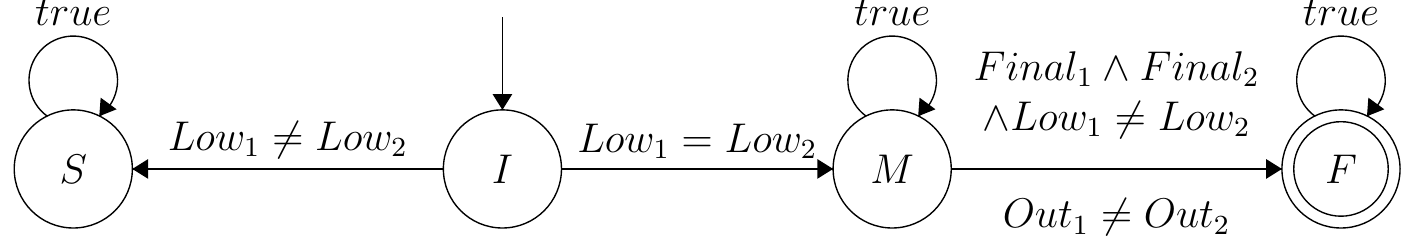}
  \end{center}
  \caption{A B\"uchi automaton for the negation of the non-interference property.}
  \label{fig:non-interference}
\end{figure}

Define the automaton-based relation $R$ by $(q,t_0,t_1) R (p,s_0,s_1)$ if $p=q$ and $t_0 \prec s_0$ and $t_1 \prec s_1$. We have the following theorem.

\begin{theorem}\label{thm:non-interference}
  If the simulation relation $\prec$ between $C$ and $S$ satisfies the additional properties needed to preserve non-interference, then $R$ is a refinement.
\end{theorem}

\begin{proof}
  Consider $(q,t_0,t_1) R (p,s_0,s_1)$. As $\prec$ is a simulation, for any joint transition from $(t_0,t_1)$ to $(t'_0,t'_1)$, there is a joint transition from $(s_0,s_1)$ to $(s'_0,s'_1)$ such that $t'_0 \prec s'_0$ and $t'_1 \prec s'_1$ holds. This transition preserves input and output values, as $\prec$ is an input-output preserving simulation.

  We have only to establish that the automaton transitions also match up. If the automaton state is either $F$ or $S$, the resulting state is the same, so by the refinement relation, we have $p'=p=q=q'$.

  Consider $q=I$. If $q'=S$ then the values of the low variables in $t_0,t_1$ differ; in which case, by condition (b), those values differ in $s_0,s_1$ as well, so $p'$ is also $S$. Similarly, if $q'=M$, then $p'=M$.

  Consider $q=M$. If $q'=F$ then either (1) $t_0,t_1$ are both final states and the values of the low variables differ; in which case, by condition (b), those values differ in $s_0,s_1$ as well, so $p'$ is also $F$, or (2) the outputs of the transitions from $t_0,t_1$ to $t'_0,t'_1$ differ; in which case, as $\prec$ preserves outputs, this is true also of the transition from $s_0,s_1$ to $s'_0,s'_1$, so $p'$ is also $F$. If $q'=M$ then the outputs are identical and one of $t_0,t_1$ is non-final; in which case, by condition (a), that is true also for the pair $s_0,s_1$, so $p'$ is also $M$.

  Finally, by the relation $R$, if $q=F$, the accepting state, then $p=F$ as well. This completes the case analysis and the proof. 
  \qed
\end{proof}

\section{Witnessing General Security Properties} \label{sec:general-witness}

The notion of refinement  presented in Section~\ref{sec:refinement} suffices for universal hyperproperties, as in that case a violation corresponds to a bundle of traces rejected by the automaton. Although many important hyperproperties are universal in nature, there are cases that require quantifier alternation. One example is \emph{generalized noninterference}, as formalized in~\cite{DBLP:conf/post/ClarksonFKMRS14}, which says that for every two traces of a program, there is a third trace that has the same high inputs as the first but is indistinguishable from the second to a low-clearance individual.
A violation for such hyperproperties, as defined in Section~\ref{sec:hyperproperties}, is not simply a bundle of traces, but rather a winning strategy for the antagonist in the corresponding game. A refinement relation does not suffice to match winning strategies. Therefore, we also introduce an input-equivalent bisimulation $B$ from $T$ to $S$, which is used in a back-and-forth manner to construct a matching winning strategy for the antagonist in $\game(S,\varphi)$ from any winning strategy for the antagonist in $\game(T,\varphi)$.

  A bisimulation $B$ ensures, by induction, that any infinite execution in $T$ has an input-equivalent execution in $S$, and vice-versa. For an execution $x$ of $T$, we use $B(x)$ to denote the set of input-equivalent executions in $S$ induced by $B$, which is non-empty. The symmetric notion, $B^{-1}(y)$, refers to input-equivalent executions in $T$ induced by $B$ for an execution $y$ of $S$. 
  
  \begin{definition} \label{def:bistrategy}
  Let $\xi^T$ be a strategy for the antagonist in $\game(T, \varphi)$ and $B$ be a bisimulation between $T$ and $S$. Then, the strategy $\xi^S = \strat(\xi^T, B)$ for the antagonist in $\game(S, \varphi)$ proceeds in the following way to produce a play $(y_1, \ldots, y_k)$:
  \begin{itemize}
  \item For every $i$ such that $\pi_i$ is existentially quantified, let $y_i$ be chosen by the protagonist in $\game(S,\varphi)$. Choose an input-equivalent execution $x_i$ from $B^{-1}(y_i)$;
  \item For every $i$ such that $\pi_i$ is universally quantified, choose $x_i$ in $\xi^T(x_1,\ldots,x_{i-1})$ and choose $y_i$ from $B(x_i)$.
  \end{itemize}
  \end{definition}

  Thus, the bisimulation helps define a strategy $\xi^S$ to match a winning antagonist stategy $\xi^T$ in $T$. We can establish that this stategy is winning for the antagonist in $S$ in two different ways. First, we do so under the assumption that $S$ and $T$ are \emph{input-deterministic}, i.e., any two executions of the program with the same input sequence have the same observation sequence. This is a reasonable assumption, covering sequential programs with purely deterministic actions.

\begin{theorem} \label{thm:typeI-refinement}
Let $S$ and $T$ be \emph{input-deterministic} programs over the same input alphabet $\pinputs$. Let $\varphi$ be a general security property with automaton $A$ representing the negation of its kernel $\kappa$. If there exists (1) a bisimulation $B$ from $T$ to $S$, and (2) a refinement relation $R$ from $A \times T^k$ to $A \times S^k$ for $\pinputs$, then $T$ securely refines $S$ for $\varphi$. 
\end{theorem}

\begin{proof}
  We have to show, from Definition~\ref{def:preservation}, that for any winning strategy $\xi^T$ for the antagonist in $\game(T,\varphi)$, there is a matching winning stategy $\xi^S$ in $\game(S,\varphi)$. Let $\xi^S = \strat(\xi^T, B)$. Let $y=(y_1,\ldots,y_k)$ be the bundle of executions resulting from a play following the strategy $\xi^S$, and $x=(x_1,\ldots,x_k)$ the corresponding bundle resulting from $\xi^T$. By construction, $y$ and $x$ are input-equivalent.
  
  Since $\xi^T$ is a winning strategy, the trace of $x$ is accepted by $A \times T^k$. Then, from the refinement $R$ and Lemma~\ref{lemma:typeI-refinement}, there is a bundle $z = (z_1, \ldots, z_k)$ accepted by $A \times S^k$ that is input-equivalent to $x$. Therefore, $z$ is a win for the antagonist. Since $z$ is input-equivalent to $x$, it is also input-equivalent to $y$. Input-determinism requires that $z$ and $y$ are identical, so $y$ is also a win for the antagonist. Thus, $\xi^S$ is a winning strategy for the antagonist in $\game(S,\varphi)$. 
  \qed
\end{proof}

If $S$ and $T$ are not input-deterministic, a new notion of refinement is defined that intertwines the automaton-based relation, $R$, with the bisimulation, $B$. 
A relation $R \subseteq (Q^A \times (\pstates^T)^k)  \times (Q^A \times (\pstates^S)^k)$ is a \emph{refinement relation} from $A \times T^k$ to $A \times S^k$ for $I$ \emph{relative to} $B \subseteq \pstates^T \times \pstates^S$, if 
\begin{enumerate}
\item $((\iota^A, \pinit^{T^k}), (\iota^A,\pinit^{S^k}))$ is in $R$ and $(\pinit^{T^k}_i,\pinit^{S^k}_i) \in B$ for all $i$; and 
\item If $((q,t),(p,s))$ is in $R$, $(t_i,s_i)$ is in $B$ for all $i$, $((q,t), v, (q',t'))$ is in $\Delta^{A\times T^k}$, $(s,u,s')$ is in $(\ptrans^{S^k})$, $u$ and $v$ agree on $I$, and $(t'_i,s'_i) \in B$, there is $p'$ such that all of the following hold:
  \begin{enumerate}
  \item $((p,s),u,(p',s')) \in \Delta^{A \times S^k}$;
  \item $((q',t'), (p',s')) \in R$;
  \item if $q' \in F$ then $p' \in F$.
  \end{enumerate}
\end{enumerate}

Typically, a refinement relation implies, as in Lemma~\ref{lemma:typeI-refinement}, that a run in $A \times T^k$ is matched by some run in $A\times S^k$. The unusual refinement notion above instead  considers already matching executions of $T$ and $S$, and formulates an inductive condition under which a run of $A$ on the $T$-execution is matched by a run on the $S$-execution. The result is the following theorem, establishing the new refinement rule, where the witness is the pair $(R,B)$.

\begin{theorem} \label{thm:typeII-refinement}
  Let $S$ and $T$ be programs over the same input alphabet $\pinputs$. Let $\varphi$ be a general security property with automaton $A$ representing its kernel $\kappa$. If there exists (1) a bisimulation $B$ from $T$ to $S$, and (2) a relation $R$ from $A \times T^k$ to $A \times S^k$ \emph{that is a refinement relative to $B$}, then $T$ securely refines $S$ for $\varphi$. 
\end{theorem}

\subsection{Checking General Refinement Relations}

The main difference when checking security preservation of general hyperproperties, compared to the purely-universal properties handled in Section~\ref{sec:refinement}, is the necessity of the compiler to provide also the bisimulation $B$ as part of the witness. The verifier must then check also that $B$ is a bisimulation, which can be performed inductively using SMT queries in a similar way to the refinement itself. In the case that the language semantics guarantee input-determinism, as described above, then Theorem~\ref{thm:typeI-refinement} holds and checking $B$ and $R$ separately is sufficient. Otherwise, the check for $R$ described in Section~\ref{sec:smt-check} has to be modified to follow Theorem~\ref{thm:typeII-refinement} by checking if $R$ is a refinement relative to $B$.

Note that appropriate formulas $B(t, s)$ describing bisimulations can be extracted from the examples of refinement relations given in Section~\ref{sec:refinement-check}:

\begin{enumerate}
    \item \textbf{Constant Folding:} $(t = s) \land (\ell^T = \texttt{L3} \rightarrow \alpha^T(\texttt{y}) = 42) \land (\ell^T = L4 \rightarrow \alpha^T(\texttt{z}) = 1) \land (\ell^T = \texttt{L5} \rightarrow \alpha^T(\texttt{x}) = 0)$
    \item \textbf{Common-Branch Factorization:} $(t = s) \lor ((\ell^T = \texttt{L2}) \land ((\alpha^S(\texttt{i}) < \alpha^S(\texttt{arr\_size}))\,?\,(\ell^S = \texttt{L2}) : (\ell^S = \texttt{L5})) \land (\alpha^T = \alpha^S[\texttt{a} := \texttt{arr[0]}]))$
    \item \textbf{Switching Instructions:} $(\ell^S = \ell^T) \land (\ell^S \neq \texttt{L6} \rightarrow \alpha^S = \alpha^T) \land (\ell^S = \texttt{L6} \rightarrow \alpha^S[\texttt{b[j]} := \texttt{a[j-1]}] = \alpha^T[\texttt{a[j]} := \texttt{b[j-1]}])$
\end{enumerate}

Note the similarities between the bisimulations above and the refinement relations from Section~\ref{sec:refinement-check}. When the transformation does not alter the observable behavior of a program, it is often the case that the refinement relation between $\buff{A} \times T^k$ and $\buff{A} \times S^k$ is essentially formed by the $k$-product of a bisimulation between $T$ and $S$ across the several tracks.

\section{Discussion and Related Work}

This work tackles the important problem of ensuring that optimizations carried out by a compiler do not break vital security properties of the source program. We propose a methodology based on property-specific refinement rules, with the refinement relations (witnesses) being generated at compile time and validated independently by a generic refinement checker. This structure ensures that neither the code of the compiler nor the witness generator have to be formally verified in order to obtain a formally verifiable conclusion. It is thus eminently suited to production compilers, which are large and complex, and are written in hard-to-formalize languages such as \verb|C| or \verb|C++|. 

The refinement rules are synthesized from an automaton-theoretic definition of a security property. This construction applies to a broad range of security properties, including those specifiable in the HyperLTL logic~\cite{DBLP:conf/post/ClarksonFKMRS14}. When applied to automaton-based formulations of the non-interference and constant-time properties, the resulting proof rules are essentially identical to those developed in the literature in~\cite{amorim-et-al-2014,deng-namjoshi-SAS-2016,DBLP:conf/eurosp/MurraySE18} for non-interference and in~\cite{DBLP:conf/csfw/BartheGL18} for constant-time. Manna and Pnueli show in a beautiful paper~\cite{DBLP:conf/tls/MannaP87} how to derive custom  proof rules for deductive verification of a LTL property from an equivalent B\"uchi automaton; our constructions are inspired by this work.

Refinement witnesses are in a form that is composable: i.e., for a security property $\varphi$, if $R$ is a refinement relation establishing a secure transformation from $A$ to $B$, while $R'$ witnesses a secure transformation from $B$ to $C$, then the relational composition $R;R'$ witnesses a secure transformation from $A$ to $C$. Thus, by composing witnesses for each compiler optimization, one obtains an end-to-end witness for the entire optimization pipeline. 

Other approaches to secure compilation include full abstraction, proposed in~\cite{DBLP:conf/ecoopw/Abadi99} (cf.~\cite{secure-compilation-survey-2019}), and trace-preserving compilation~\cite{DBLP:conf/csfw/PatrignaniG17}. These are elegant formulations but difficult to check fully automatically, and hence not suitable for translation validation. The theory of hyperproperties~\cite{DBLP:conf/csfw/ClarksonS08} includes a definition of refinement in terms of language inclusion (i.e., $T$ refines $S$ if the language of $T$ is a subset of the language of $S$), which implies that any subset-closed hyperproperty is preserved by refinement in that sense. Language inclusion is also not directly checkable and thus cannot be used for translation validation. The refinement theorem in this paper for universal properties (which are subset-closed) uses a tighter step-wise inductive check that is suitable for automated validation.

Translation validation through compiler-generated refinement relations  arises from work on ``Credible Compilation'' by~\cite{rinard99,darkothesis} and ``Witnessing'' by~\cite{namjoshi-zuck-2013}. As the compiler and the witness generator do not require formal verification, the size of the trusted code base shrinks substantially. Witnessing also requires much less effort than a full mathematical proof: as observed in~\cite{namjoshi2014witnessing}, a mathematical correctness proof of SSA (Static Single Assignment) conversion in Coq is about 10,000 lines~\cite{DBLP:conf/pldi/ZhaoNMZ13}, while refinement checking can be implemented in around 1,500 lines of code, most of which comprises a witness validator which can be reused across different transformations. Our work shows how to extend this concept, originally developed for correctness checking, to the preservation of a large class of security properties, with the following important distinction. The refinement relations used for correctness are strong in that (via well-known results) refinement preserves \emph{all} linear-time properties defined over atomic propositions common to both programs. Strong preservation is needed as the desired correctness properties may not be fully known or their specifications may not be available in practice. On the other hand, security properties are limited in nature and are likely to be well known in advance (e.g., ``do not leak secret keys''). This motivates our construction of property-specific refinement relations. Being more focused, they are also easier to establish than refinements that preserve \emph{all} properties in a class.

The refinement rules defined here implicitly require that a security specification apply equally well to the target and source programs. Thus, they are most applicable when the target and source languages and attack models are identical. This is the case, for instance, in the optimization phase of a compiler, where a number of transformations are applied to code that remains within the same intermediate representation. To complete the picture, it is necessary to look more generally at transformations that go from a higher-level language (say LLVM bytecode) to a lower-level one (say \verb|x86| machine code). The so-called ``attack surfaces'' are different for these two levels, so it would be necessary to also incorporate a \emph{back-translation} of failures~\cite{DBLP:conf/popl/DevriesePP16} in the refinement proof rules. How best to do so is an intriguing topic for future work. 

Another question that we leave to future work is the completeness of the refinement rules. We have shown that a variety of common compiler transformations can be proved secure through logically simple refinement relations. The completeness question is whether \emph{every} secure transformation has an associated stepwise refinement relation. In the case of correctness, this is a well-known theorem by Abadi and Lamport~\cite{DBLP:conf/lics/AbadiL88}. To the best of our knowledge, a corresponding theorem is not known for security hyperproperties. 

There are a number of practical concerns that must be addressed to implement this methodology in a real compiler. One of these is a convenient notation for specifying desired security properties at the source level, for example as annotations to the source program. It is also necessary to define precisely how a security property is transformed by a program optimization. For instance, if a transformation introduces fresh variables, there needs to be a reasonable way to determine whether those should be assigned a high or low security level for a non-interference property. 

\bibliographystyle{plain}
\bibliography{main}

\newpage
\appendix

\section{Appendix}

\newtheorem{numbered-lemma}{Lemma}

\newtheorem{numbered-theorem}{Theorem}

\subsection{Proofs}

\setcounter{numbered-lemma}{1}

\begin{numbered-lemma}
Let $S,T$ be programs over the same alphabet $\EXT{\Sigma}$ and $A$ be a B\"uchi automaton over $\EXT{\Sigma}^k$. Let $I$ be a subset of $\EXT{\Sigma}$, and let $R$ be a refinement relation from $A \times T^k$ to $A \times S^k$ for $I$. For every sequence $v$ in $\Lang(A \times T^k)$, there is a  sequence $u$ in $\Lang(A \times S^k)$ such that $u$ and $v$ are input-equivalent.
\end{numbered-lemma}

\begin{proof}
Let $v = v_0, v_1 \ldots$ be in $\Lang(A \times T^k)$. Let $\rho^T = (q_0=\iota^A,t_0=\iota^{T^k}),v_0, (q_1,t_1), v_1, \ldots$ be an accepting run on this sequence. By a simple induction from the refinement conditions, there is a sequence $u=u_0, u_1, \ldots$ and a run $\rho^S = (p_0=\iota^A,s_0=\iota^{S^k}),u_0,(p_1,s_1),u_1, \ldots$ of $A \times S^k$ on $u$ such that for each $i$, all of the following hold. 
\begin{enumerate}
\item $((q_i,t_i), (p_i,s_i)) \in R$;
\item $u_i$ and $v_i$ agree on $I$;
\item if $q_i \in F$ then $p_i \in F$.
\end{enumerate}

As the run $\rho^T$ is accepting for $A\times T^k$, there are infinitely many points on the run satisfying $F$. By the third condition above, $F$ holds infinitely often along $\rho^S$ as well, so that run is accepting for $A \times S^k$. By the second condition, $u_i$ and $v_i$ agree on $I$ for all $i$, so $u$ and $v$ are input-equivalent, as required.
\qed
\end{proof}

\setcounter{numbered-theorem}{5}

\begin{numbered-theorem}
Let $S$ and $T$ be programs over the same input alphabet $\pinputs$. Let $\varphi$ be a general security property with automaton $A$ representing its kernel $\kappa$. If there exists (1) a bisimulation $B$ from $T$ to $S$, and (2) a relation $R$ from $A \times T^k$ to $A \times S^k$ \emph{that is a refinement relative to $B$}, then $T$ securely refines $S$ for $\varphi$. 
\end{numbered-theorem}

\begin{proof}

  Let $\xi^T$ be a winning strategy for the opponent in $\game(T, \varphi)$, and $\xi^S = \strat(\xi^T, B)$. By construction, a play of $\xi^S$ results in a bundle of $k$ executions of $S$, $y=(y_1,\ldots,y_k)$, such that the bundle is pointwise input-equivalent to a bundle $x=(x_1,\ldots,x_k)$ resulting from a play following $\xi^T$. Furthermore, by construction, $x$ and $y$ are pointwise related by $B$.

  As $\xi^T$ is a winning strategy for the opponent, $(\trace(x_1),\ldots,\trace(x_k))$ does not satisfy $\kappa$. Therefore, $w = \zip(\trace(x_1), \ldots, \trace(x_k))$ has an accepting run $r^T$ on $A \times T^k$ for which the program components are $\zip(\states(x_1), \ldots, \states(x_k))$.
  
  By the definition of $R$, the initial states of $A \times T^k$ and $A \times S^k$ are related by $R$, and the program components of these states are related pointwise by $B$. These program components are also the initial states of $\zip(x_1, \ldots, x_k)$ and $\zip(y_1, \ldots, y_k)$. By construction, the $i$-th entry on $\zip(x_1,\ldots,x_k)$ is related by $B^k$ to the $i$-th entry on $\zip(y_1,\ldots,y_k)$, for all positions $i$. By an inductive argument using the refinement relation, one can construct a corresponding run, $r^S$, of $A \times S^k$ such that the program components are $\zip(\states(x_1), \ldots, \states(x_k))$. Moreover, for every final state in $r^T$, its corresponding state in $r^S$ is also final. As $r^T$ is accepting, so is $r^S$. Hence, the play $y$ is a win for the opponent.

  As this is true for arbitrary choices made by the player in $\game(S,\varphi)$, the strategy $\xi^S$ is a winning strategy, as required.
  \qed
\end{proof}

\subsection{Buffering Automata} \label{sec:buffering-automata}

The refinement relation described in Definition~\ref{def:refinement} makes use of a version of automaton-program product that assumes that every program and automaton transition has an observable label. Although this can be reasonable depending on the attack model (for example, when modeling constant-time security~\cite{DBLP:conf/csfw/BartheGL18}), in many cases the program will include $\emptystring$-transitions that produce no observation. This can lead to situations when two output traces become unsynchronized. At the end of Section~\ref{sec:refinement} we briefly described a way to handle this issue by constructing a \emph{buffering automata}. We now explain this construction in more detail.

As an example of when buffering automata are needed, consider the program below, and the property which says that the output trace is identical for any pair of inputs. 
\begin{verbatim}
  int x := read_input();
  for(i := 0; true; i := i+1) {
    write_output(i);
    for(int t := x; t > 0; t := t-1); // delay for x steps
  }
\end{verbatim}

It is evident that this property holds, as every execution produces the output sequence $0,1,\ldots$. Yet, successive outputs are separated by $|\texttt{x}|$ steps, thus, outputs of executions with different input values are not synchronized, and the gap between corresponding outputs grows without bound.

To match this program behavior to the automaton for the negated trace property, we require converting the automaton $A$ over $\EXT{\Sigma}^k$ into an automaton $\buff{A}$ over $(\EXT{\Sigma} \cup \{\emptystring\})^k$, called a \emph{buffering automaton}, with the following property:
  $\buff{A}$ accepts $\zip(v_1,\ldots,v_k)$ iff $A$ accepts $\zip(\compress(v_1),\ldots,\compress(v_k))$. 
This produces an infinite-state automaton with buffers of unbounded size, that store observation histories of each trace. As the proof method is based on single-step refinement, an infinite-state automaton is not necessarily an obstacle to automated checking, so long as its structure can be represented symbolically in an SMT-supported theory. Furthermore, in many cases the refinement check can be performed without needing to reason about the structure of the automaton at all (see Section~\ref{sec:refinement-check} for examples).

Formally, given an automaton $A = (Q, \Sigma^k, \iota, \Delta, F)$, we define the \emph{buffering automaton} $\buff{A} = (Q \times (\Sigma^*)^k \times \{0, 1\}, (\Sigma \cup \{\emptystring\})^k, (\iota, \emptystring, \ldots, \emptystring, 0), \buff{\Delta}, F \times (\Sigma^*)^k \times \{1\})$. Note that the state space of $\buff{A}$ stores $k$ finite sequences of observations, representing the prefix of the traces produced so far by each copy of the program. The alphabet of $\buff{A}$ allows the copies to perform $\emptystring$-transitions, representing the fact that each track might produce observations at different rates.
The transition relation $\buff{\Delta}$ is defined in the following way. $((q, w, b), u, (q', w', b')) \in \buff{\Delta}$, for $q, q' \in Q$, $w, w' \in (\Sigma^*)^k$, $u \in (\Sigma \cup \{\emptystring\})^k$ and $b, b' \in \{0, 1\}$, if:

\begin{enumerate}
\item $q = q'$, $w_i u_i = w'_i$ for all $i \in \{1, \ldots, k\}$, $w'_i = \emptystring$ for some $i$, and $b' = 0$; or
\item there is $\sigma \in \Sigma^k$ such that $w_i u = \sigma_i w'_i$ for all $i \in \{1, \ldots, k\}$, $(q, \sigma, q') \in \Delta$, and $b' = 1$.
\end{enumerate}

Intuitively, whenever the $k$ copies of the program perform a transition, producing observations $u = (u_1, \ldots, u_k)$, $\buff{A}$ concatenates $u_i$ with the sequence $w_i$ stored as part of the state. Then, if any of the resulting sequences is empty, the state $q$ of the automaton does not advance. Otherwise, if all of the sequences are non-empty, the automaton reads the first position of all sequences, removing them from the buffer, and transitions to a new state $q'$ accordingly. In this way, $\buff{A}$ visits the same states that $A$ would visit if all of the tracks were synchronized, possibly with delays due to having to wait until all tracks have produced the next observation.

The $\{0, 1\}$ component of the state is needed to indicate whether the automaton makes progress. It remains as $0$ while the automaton is still waiting for an observation from one of the tracks, and is set to $1$ whenever the automaton takes a transition to a different $Q$ state. $\buff{A}$ accepts iff it makes progress into an accepting $Q$ state infinitely often.

\begin{lemma} \label{lemma:buffers}
$\buff{A}$ accepts $\zip(v_1, \ldots, v_k)$ iff A accepts $\zip(\compress(v_1), \ldots, \compress(v_k))$.
\end{lemma}
\begin{proof}
Given an accepting run $\buff{r}$ of $\buff{A}$ on $\zip(v_1, \ldots, v_k)$, a corresponding accepting run $r$ of $A$ on $\zip(\compress(v_1), \ldots, \compress(v_k))$ can be constructed in the following way. First, note that there are two kinds of transitions in $\buff{r}$, each corresponding to one of the two cases in the definition of $\buff{\Delta}$. Due to our assumption that no program has a trace that after a point only produces $\emptystring$, transitions of the second kind must occur infinitely often in $\buff{r}$. For every such transition $((q, w, b), u, (q', w', b'))$ there is a corresponding transition $(q, \sigma, q') \in \Delta$, and between two such transitions $q$ does not change. Therefore, the corresponding transitions in $\Delta$ can be put end-to-end to form a run $r$ on $A$. By construction of $\buff{\Delta}$, the trace of $r$ is precisely $\zip(\compress(v_1), \ldots, \compress(v_k))$. Furthermore, note that accepting states $(q', w', 1)$ of $\buff{A}$ can only be reached by transitions of the second kind, and in every such state $q'$ is an accepting state of $A$. Therefore, every occurrence of an accepting state $(q', w', 1)$ in $\buff{r}$ has a corresponding occurrence of $q'$ in $r$. Since such occurrences happen infinitely often, $r$ is an accepting run.

To prove the opposite direction, assume that $r$ is an accepting run of $A$ on $\zip(\compress(v_1), \ldots, \compress(v_k))$. Then, a corresponding accepting run $\buff{r}$ of $\buff{A}$ on $\zip(v_1, \ldots, v_k)$ can be constructed in the following way. Note that the only source of nondeterminism in $\buff{\Delta}$ is in the choice of $(q, \sigma, q') \in \Delta$ for the second kind of transition. If this choice is always made by choosing the next transition in $r$, then an accepting state $(q', w', 1)$ will be visited infinitely often. Therefore, $\buff{r}$ will be accepting. The construction of $\buff{\Delta}$ guarantees that $\buff{r}$ constructed in this way will be a valid run of $\zip(v_1, \ldots, v_k)$ in $\buff{A}$.
\qed
\end{proof}

Since the alphabet of $\buff{A}$ is $(\Sigma \cup \{\emptystring\})^k$, we can take its product with the $k$-fold composition of a program with $\emptystring$ transitions.

\begin{lemma}
  A bundle of executions $\sigma=(\sigma_1,\ldots,\sigma_k)$ produced by program $P$ does not satisfy $\kappa$ iff $\zip(\trace(\sigma_1), \ldots, \trace(\sigma_k))$ is accepted by $\hat{A} \times P^k$. 
\end{lemma}
\begin{proof}
  Bundle $\sigma$ does not satisfy $\kappa$
  iff (by definition) $\zip(\ctrace(\sigma_1),\ldots,\ctrace(\sigma_k))$ is accepted by $A$,
  iff (by Lemma~\ref{lemma:buffers}) $\zip(\trace(\sigma_1),\ldots,\trace(\sigma_k))$ is accepted by $\hat{A}$,
  iff (as $\zip(\trace(\sigma_1),\ldots,\trace(\sigma_k))$ is the trace of execution $\zip(\sigma_1,\ldots,\sigma_k)$ of $P^k$ and Lemma~\ref{lemma:product}) $\zip(\trace(\sigma_1),\ldots,\trace(\sigma_k))$ is accepted by $\hat{A} \times P^k$.
  \qed
\end{proof}

\subsection{Additional Examples of Refinement Relations}

\subsubsection{Example 3: Switching Instructions (Special Cases)}

As mentioned in Section~\ref{sec:refinement-check}, setting up a refinement relation for the switching-instructions transformation becomes harder if the instructions being switched produce observations, due to the fact that $(q^S = q^T)$ might not hold at the location of the switch. There are some special cases, however, where this may be done more easily:

\begin{enumerate}
    \item If the different copies of $S$ and $T$ are guaranteed to be synchronized (i.e., it is always the case that $\ell^S_i = \ell^S_j$ and $\ell^T_i = \ell^T_j$, for all $i$ and $j$).
    \item If only one of the switched instructions produces an observation.
\end{enumerate}

The advantage of the first case is that if all $k$ copies of the program produce corresponding observations at the same time, the automaton can compare them immediately and does not need to store them in its internal state. Then, it is likely that for many relevant security properties the state of the automaton does not change at all when executing the switched instructions, and so the relationship $(q^S = q^T)$ still holds. In this case the constraint that $\ell^S_i = \ell^S_j$ and $\ell^T_i = \ell^T_j$, for all $i$ and $j$, must be added to $R$. Additionally, $R$ will need to include enough information for the inductive check to confirm that the state of the automaton indeed does not change while observing the switched instructions. What this information is will depend on the property under verification.

The second case can be handled in a couple of different ways. One option, for example, is to relax the constraint $(\ell^S_i = \ell^T_i)$ to instead relate the instruction that produces an observation in the source with the same instruction that produces an observation in the target (say, $L5$ in the source with $L6$ in the target). Then, stuttering can be used in both the source and the target to align the executions (for example, the target stutters in $L5$, and the source stutters in $L6$). This will also require changing the constraints $(\ell^S_i \neq \ell^S_{switch} \rightarrow \alpha^S_i = \alpha^T_i) \land (\ell^S_i = \ell^S_{switch} \rightarrow \delta(\alpha^S_i, \alpha^T_i))$ keeping track of the changes in state appropriately. The modified refinement seeks to align the observations rather than the instructions.
An easier alternative, instead, might be to model $S$ and $T$ in such a way that the two instructions are considered as a block, comprising a single transition in the transition systems. Then, since the two blocks modify the state of the program in the same way and produce the same observation, the refinement relation becomes trivial.

\subsubsection{Example 4: Dead-Branch Elimination}

This optimization replaces an \emph{if} statement whose guard is trivially false by only the contents of the \emph{else} branch. For simplicity, assume there is always an \emph{else} branch (\emph{if} statements with no \emph{else} branch can be modeled by leaving the \emph{else} branch empty). Consider the following example programs:

\begin{multicols}{2}

\begin{verbatim}
// Source program S
L1: int x := secret_input();
L2: if (0) {
L3:     public_output(x);
L4: } else {
L5:     secret_output(x);
L6: }
\end{verbatim}

\columnbreak

\begin{verbatim}
// Target program T
L1: int x := secret_input();
L2: secret_output(x);
\end{verbatim}

\end{multicols}

The refinement relation for this source-target pair can be written as:
\begin{align*}
    (q^S = q^T) \land \bigwedge^k_{i=1} (\alpha^S_i = \alpha^T_i) \land L(\ell^S_i, \ell^T_i)
\end{align*}
where $L = \{(L1, L1), (L2, L2), (L5, L2), (End, End)\}$ denotes which locations can be related between the source and the target. Note that since $L2$ in the target is related to both $L2$ and $L5$ in the source, it is necessary to allow stuttering when checking the refinement relation.

This same refinement relation can be used for arbitrary programs, as long as $L$ is redefined appropriately. Let $\ell^S_{if}$ be the program location of the \emph{if} guard ($L2$ in the example), and $\ell^S_{else}$ be the program location at the start of the \emph{else} branch ($L5$ in the example). Let $\ell^T_{else}$ be the location in the target corresponding to $\ell^S_{else}$ ($L2$ in the example). $L$ is then defined to be a binary relation relating program locations between the source and target in the following way: corresponding locations are related by $L$, and $\ell^T_{else}$, in addition to being related to $\ell^S_{else}$, is also related to $\ell^S_{if}$.

Note that this relation also defines a relation $B((\alpha^S_i, \ell^S_i), (\alpha^T_i, \ell^T_i)) = (\alpha^S_i = \alpha^T_i) \land L(\ell^S_i, \ell^T_i)$ between single program states of $S$ and $T$ such that $B$ is a (stuttering) bisimulation. Therefore, this refinement relation can also be used for security properties with arbitrary quantifier alternation.

\subsubsection{Example 5: Expression Flattening}

This optimization ``flattens'' a nested expression. It can be thought of as turning an expression into three-address code. For example:

\begin{multicols}{2}

\begin{verbatim}
// Source program S
L1: int x := (f(a) + b) * (g(c) / d);
\end{verbatim}

\columnbreak

\begin{verbatim}
       // Target program T
       L1: int t0 := f(a);
       L2: int t1 := t0 + b;
       L3: int t2 := g(c);
       L4: int t3 := t2 / d;
       L5: int x := t1 * t3;
\end{verbatim}

\end{multicols}

Note that many programming languages leave the evaluation order of subexpressions undefined. Different orders of evaluation may generate different result states, e.g., if subexpressions such as \verb|f(a)|, \verb|g(c)| have side-effects. This transformation picks a certain evaluation order. Thus, while in a detailed program semantics the original expression has an undetermined evaluation order, a specific order is fixed in the target program. This is an example of a transformation that removes non-determinism from the program.

The original and flattened evaluations are equivalent when the detailed evaluation is treated as a block, ignoring the values of intermediate variables. Thus, the correctness argument establishes that, starting from states related by $\alpha^S_i(v) = \alpha^T_i(v)$ for all $v \in \{\texttt{a}, \texttt{b}, \texttt{c}, \texttt{d}\}$, the ordered execution in the target $T$ matches the result of one of the possible non-deterministic evaluation choices in the source $S$, resulting in states that are also identically related, in this case $\alpha^S_i(v) = \alpha^T_i(v)$ for all $v \in \{\texttt{a}, \texttt{b}, \texttt{c}, \texttt{d}, \texttt{x}\}$.

For any automaton, we thus define the relation $R$ as:
\begin{align*}
    (q^S = q^T) \land \bigwedge^k_{i=1} (\alpha^S_i =_V \alpha^T_i)
\end{align*}
where $V = \{\texttt{a}, \texttt{b}, \texttt{c}, \texttt{d}, \texttt{x}\}$ denotes all non-temporary variables (i.e., excluding variables introduced by the transformation), and $(\alpha^S_i =_V \alpha^T_i)$ denotes the pointwise equality of $\alpha^S_i(v)$ and $\alpha^T_i(v)$ for all $v \in V$.

It is easy to show that this relation is a refinement relation if the transition system representing $T$ is modeled in such a way that the transformed segments are treated as a block of instructions defining a single transition. In this case the refinement check requires only a single inductive step, following the reasoning described above. If instead the instructions in the target are treated as individual transitions rather than a block, it is necessary to add constraints to remember the value of the intermediate variables in the target, for example $(\ell^T_i \in \{L3, L4, L5\} \rightarrow \alpha^T_i(\texttt{t1}) = \alpha^T_i(\texttt{t0}) + \alpha^T_i(\texttt{b}))$. With this addition, the relation can then be used as a stuttering refinement, where the stuttering is bounded by the length of the expanded block.

Note that the correctness of the refinement relies on the assumption that there is no change in observations between the original single-instruction block and the new multiple-instruction block. That may not be the case depending on the attack model. For instance, if \verb|x := (secret + 2) - secret| is compiled to \verb|t1 := secret + 2; x := t1 - secret| then the secret value is leaked in the value of \verb|t1|, even though there is no leakage in the original assignment to \verb|x| in the source. Whether this matters depends on the attack model. If the attack model allows values of temporary variables to be observable, then the transformation is not secure.

Unlike previous examples, the correspondence for single traces in expression flattening is not in general a bisimulation, since the elimination of non-determinism means that there might be a transition in the source with no corresponding transition in the target. Therefore, the security preservation result applies only to universally-quantified properties.

\subsubsection{Example 6: Loop Peeling}

This optimization peels off the first iteration of a loop, as long as the compiler can guarantee that the loop guard will be true on entry. For simplicity, assume only \emph{while} loops. Consider the following example programs:

\begin{multicols}{2}

\begin{verbatim}
    // Source program S
    L1: int x := 0;
    L2: int k := 0;
    L3: while (k < 8) {
    L4:     if (k == 0) {
    L5:         x := secret_input();
    L6:     } else {
    L7:         x := x + x;
    L8:     }
    L9:     k := k + 1;
    L10:    public_output(x % k);
    L11: }
\end{verbatim}

\columnbreak

\begin{verbatim}
    // Target program T
    L1: int x := 0;
    L2: int k := 0;
    L3: if (k == 0) {
    L4:     x := secret_input();
    L5: } else {
    L6:    x := x + x;
    L7: }
    L8: k := k + 1;
    L9: public_output(x % k);
    L10: while (k < 8) {
    L11:    if (k == 0) {
    L12:        x := secret_input();
    L13:    } else {
    L14:        x := x + x;
    L15:    }
    L16:    k := k + 1;
    L17:    public_output(x % k);
    L18: }
\end{verbatim}
\end{multicols}

Since the two programs produce the same traces, the following refinement relation can be used independently of the security property:
\begin{align*}
    L = \{&(L1, L1), (L2, L2), (L3, L3), (L4, L3), (L5, L4), (L7, L6), (L9, L8), (L10, L9), \\
    &(L3, L10), (L4, L11), (L5, L12), (L7, L14), (L9, L16), (L10, L17), (End, End)\}
\end{align*}
\begin{align*}
    (q^S = q^T) \land \bigwedge^k_{i=1} (\alpha^S_i = \alpha^T_i) \land L(\ell^S_i, \ell^T_i) \land (\ell^T_i = L3 \rightarrow \alpha^S_i(k) < 8)
\end{align*}

This refinement relation requires stuttering, as the target stutters in $L3$ while the source moves from $L3$ and $L4$.
For general programs, $L$ should be defined as follows:

\begin{enumerate}
    \item Corresponding locations are related by $L$.
    \item Every instruction in the peeled loop iteration in the target is related by $L$ to the corresponding instruction within the loop body in the source.
    \item $\ell^T_{peel}$, the location at the start of the peeled loop iteration in the target, is additionally related by $L$ to $\ell^S_{guard}$, the location of the loop guard in the source. In the example, both $\ell^T_{peel}$ and $\ell^S_{guard}$ are $L3$.
\end{enumerate}

The formula $\varphi_0 = (\ell^T_i = L3 \rightarrow \alpha^S_i(k) < 8)$ in the example serves to state that the loop will be entered. It must be added to the refinement relation in order to justify the execution of the peeled iteration in the target. For arbitrary programs, the general form of $\varphi_0$ is $(\ell^T_i = \ell^T_{peel} \rightarrow \psi(\alpha^S_i))$, where $\psi$ is the loop guard at location $\ell^S_{guard}$.
In the example, adding just $\varphi_0$ suffices, since it can be proved inductively from the previous instruction $\texttt{int k := 0}$. However, in more complex cases other formulas $\varphi_1, \ldots, \varphi_n$ of the same form might need to be added to describe the compiler's reasoning leading up to the conclusion that $\varphi_0$ holds. For example, suppose that $\texttt{int k := 0}$ was replaced by $\texttt{int k := x}$. Then it would be necessary to add a formula such as $\varphi_1 = (\ell^T_i = L2 \rightarrow \alpha^S_i(x) < 8)$ in order for the relation to pass the refinement check. Therefore, the general form of the refinement relation for an arbitrary program is:
\begin{align*}
    (q^S = q^T) \land \bigwedge^k_{i=1} (\alpha^S_i = \alpha^T_i) \land L(\ell^S_i, \ell^T_i) \land \varphi_0 \land \varphi_1 \land \ldots \land \varphi_n
\end{align*}

For security properties with arbitrary quantifier alternation, the bisimulation can simply be defined as $B((\alpha^S_i, \ell^S_i), (\alpha^T_i, \ell^T_i)) = (\alpha^S_i = \alpha^T_i) \land L(\ell^S_i, \ell^T_i) \land \varphi_0 \land \varphi_1 \land \ldots \land \varphi_n$.

\subsubsection{Example 7: Register Spilling}

This transformation ``spills'' the contents of a register to memory, if the register has to be free for a subsequent operation. For example, consider a machine with only two registers, \texttt{A} and \texttt{B}.

\begin{multicols}{2}
\begin{verbatim}
// Source program S
L1: int t0 := a + b;
L2: public_output(t0);
L3: int t1 := a - b;
\end{verbatim}

\columnbreak

\begin{verbatim}
// Target program T
// assume A = a, B = b
L1: spill[0] := A;
L2: A := A + B;
// A = t0, B = b, spill[0] = a
L3: public_output(A);
L4: A := spill[0];
L5: A := A - B;
// B = b, A = t1, spill[0] = a
\end{verbatim}
\end{multicols}

There is a simple correspondence between the source variables and the target registers and spill array entries, inferred by the register allocation and spilling method, and shown in the code comments. Note that this correspondence varies depending on the program location. Therefore, it can be expressed by a mapping $\sigma : (\textsc{Reg} \cup \textsc{Spill}) \times \textsc{Loc} \rightarrow \mathcal{V}$ from a register (or position in the spill array) and a program location in the target to a corresponding variable in the source.

The two programs can thus be related by the following stuttering refinement relation $R$ (note that the stuttering can be avoided if we allow transitions composed of blocks of instructions):
\begin{align*}
    L = \{(L1, L1), (L1, L2), (L2, L3), (L3, L4), (L3, L5), (End, End)\}
\end{align*}
\begin{align*}
    (q^S = q^T) \land \bigwedge^k_{i=1} \left(L(\ell^S_i, \ell^T_i) \land \bigwedge_{v \in (\textsc{Reg} \cup \textsc{Spill})} (\alpha^S_i(\sigma(v, \ell^T_i)) = \alpha^T_i(v))\right)
\end{align*}

The refinement ensures that at every point of the program, every register and position in the spill array contains the same value stored in the corresponding variable of the source given by $\sigma$. The single-trace correspondence $(L(\ell^S_i, \ell^T_i) \land \bigwedge_{v \in (\textsc{Reg} \cup \textsc{Spill})} (\alpha^S_i(\sigma(v, \ell^T_i)) = \alpha^T_i(v))$ is a stuttering bisimulation, therefore this refinement can also be used with properties with arbitrary quantifier alternation.

Note that we are making the assumption that the transformation preserves observations. This might not be the case if the attack model includes, for example, the address of memory positions accessed by the program, as the target program introduces a new array for spilling. In this case security might be preserved or not, depending on the security property under verification, and a specialized refinement relation would be necessary.

\end{document}